\documentclass[ahp]{birkmult}

 \newtheorem{thm}{Theorem}[section]
 
 \newtheorem{lem}[thm]{Lemma}
 
 \theoremstyle{definition}
 \newtheorem{defn}[thm]{Definition}
 \theoremstyle{remark}

 \numberwithin{equation}{section}

\begin{document}

%
%
%
%
%
%
%
%
%
\title[KS Sets \&\ $n$OA Equations]
 {Kochen-Specker Sets and\\ Generalized Orthoarguesian Equations}
\author[Megill]{Norman D.~Megill}

\address{%
Boston Information Group\\
19 Locke Ln.\\
Lexington, MA 02420\\
USA}

\email{nm@alum.mit.edu}

\author[Pavi\v ci\'c]{Mladen Pavi\v ci\'c}
\address{ Institute for Theoretical Atomic, Molecular, and Optical Physics\br
Physics Department at Harvard University and\br
Harvard-Smithsonian Center for Astrophysics\br
 Cambridge, MA 02138 USA and\br
Physics Chair\br
Faculty of Civil Engineering\br
University of Zagreb\br
Zagreb, Croatia}
\email{pavicic@grad.hr}


\date{June 25, 2009}

\begin{abstract}
Every set (finite or infinite) of quantum vectors (states) satisfies
generalized orthoarguesian equations ($n$OA). We consider two
3-dim Kochen-Specker (KS) sets of vectors and show how each of
them should be represented by means of a Hasse diagram---a lattice,
an algebra of subspaces of a Hilbert space--that
contains rays and planes determined by the vectors so as to satisfy
$n$OA. That also shows why they cannot be represented by a special
kind of Hasse diagram called a Greechie diagram, as has been
erroneously done in the literature. One of the KS sets (Peres') is an
example of a lattice in which 6OA pass and 7OA fails, and that closes
an open question of whether the 7oa class of lattices properly contains
the 6oa class. This result is important because it provides additional
evidence that our previously given proof of
$n {\rm oa}\subseteq(n+1){\rm oa}$ can be extended to proper
inclusion $n {\rm oa}\subset(n+1){\rm oa}$ and that
$n {\rm OA}$ form an infinite sequence of successively stronger
equations.
 \end{abstract}

\maketitle

\section{Introduction}
\label{sec:intro}

Many authors have tried to empirically justify the mathematically
well-established orthoisomorphism between the so-called Hilbert lattice
and the lattice of subspaces of a Hilbert
space, which has been worked out by many authors over the last
75 years.\cite{birk-v-neum,beltr-cass-book,holl95}
However, a missing link between empirical quantum measurements
and its lattice structure was a proper description of a correspondence
between the standard quantum measurements, which use Hilbert
vectors and states, on the one hand, and Hilbert lattices (algebras
of the closed subspaces of Hilbert space), which make use of
subspaces that contain these vectors and/or are spanned by them,
on the other. What hampered a search for such a correspondence
was a too narrow focus on orthogonality and on infinite-dimensionality
via Greechie lattices (meaning the lattices depicted by Greechie
diagrams). In Ref.\ \cite{bdm-ndm-mp-fresl-10} we gave two
examples: empirical reconstruction of quantum mechanics via
lattice theory and a description of Kochen-Specker's setups via lattice theory.

A lattice can correctly represent a given formal description of a
quantum system only if it satisfies all the equations that the lattice of
subspaces of a Hilbert space satisfies. The only known set of equations
that are related to the algebraic structure of the latter lattice (i.e., excluding
those that are related to states introduced on the lattice) are the
generalized orthoarguesian equations ($n$OA, $n\ge 3$).\ 
\cite{mp-alg-hilb-eq-09} Thus, these equations are an essential tool for 
analysing lattices for particular experimental setups. If a lattice does not 
pass $n$OA for all $n$, then it is not a correct lattice.

In this paper, we analyze two Kochen-Specker (KS) setups: Bub's \cite{bub}
and  Peres' \cite{peres}. We represent them by MMP hypergraphs. Vectors
correspond to vertices in MMP hypergraphs, and tetrads of orthogonal
vectors correspond to edges in MMP hypergraphs. MMP hypergraphs (also
called MMP diagrams) are defined in Ref.\ \cite{pmmm04b} and in
Sec.\ \ref{sec:def1}. One can establish a correspondence between MMP
hypergraphs and lattices of subspaces of a Hilbert space. In such
lattices the vertices of MMP hypergraphs correspond to lattice atoms and
their edges to lattice blocks. Thus any KS setup can eventually be
represented by a lattice.

In Sec.\ \ref{sec:represent}, we first show why KS
setups  cannot be represented by a special kind of lattice, called
Greechie lattices, as erroneously claimed in the literature.  Then we
explain how they can be represented for use with any specifically
chosen Hilbert lattice equation.  In doing so, we introduce a new kind of
lattice---we call it MMPL---that represents all nonorthogonal lattice
elements as well as their meets and joins that take place in a
proof of the chosen equation.  As specific examples, we consider the
3OA equation for Bub's KS lattice and 7OA for Peres'.

We show that Peres' lattice satisfies 3OA through 6OA but violates
7OA. In Sec.\ \ref{sec:main}, we then generate a serious of other
lattices with this property. In Ref.\ \cite{mpoa99}, we proved
that all individual orthoarguesian equations found previously
(by other authors) were equivalent to either 3OA or 4OA and
showed lattices in which 3OA and 4OA passed but 5OA failed.
In Ref.\ \cite{pm-ql-l-hql2}, we found lattices in
which 6OA failed and OAs up to 5OA passed.

Therefore, our finding of a series of lattices that satisfy  3OA-6OA but
fail in 7OA amounts to a very strong indication that $n$oa's properly
contain each other for successively increasing $n$, for all $n$.

\section{Lattice Definitions and Theorems}
\label{sec:def1}

The closed subspaces of a Hilbert space form an algebra called a Hilbert
lattice (defined by Def.\ \ref{def:hl}).  In any Hilbert lattice,
the operation \it meet\/\rm, $a\cap b$, corresponds to
set intersection, ${\mathcal H}_a\bigcap{\mathcal H}_b$, of subspaces ${\mathcal
H}_a,{\mathcal H}_b$ of Hilbert space ${\mathcal H}$, the ordering relation
$a\le b$ corresponds to ${\mathcal H}_a\subseteq{\mathcal H}_b$, the operation
\it join\/\rm, $a\cup b$, corresponds to the smallest closed subspace of
$\mathcal H$ containing ${\mathcal H}_a\bigcup{\mathcal H}_b$, and
the \it orthocomplement\/\rm\ $a'$ corresponds
to ${\mathcal H}_a^\perp$, the set of vectors orthogonal to all vectors in
${\mathcal H}_a$. Within Hilbert space there is also an operation which
has no parallel in the Hilbert lattice: the sum of two subspaces
${\mathcal H}_a+{\mathcal H}_b$, which is defined as the set of sums of vectors
from ${\mathcal H}_a$ and ${\mathcal H}_b$. We also have
${\mathcal H}_a+{\mathcal H}_a^\perp={\mathcal H}$. One can define
all the lattice operations on a Hilbert space itself following the above
definitions (${\mathcal H}_a\cap{\mathcal H}_b={\mathcal H}_a\bigcap{\mathcal H}_b$,
etc.). Thus we have
${\mathcal H}_a\cup{\mathcal H}_b=\overline{{\mathcal H}_a+{\mathcal H}_b}=
({\mathcal H}_a+{\mathcal H}_b)^{\perp\perp}=
({\mathcal H}_a^\perp\bigcap{\mathcal
H}_b^\perp)^\perp$,\cite[p.~175]{isham} where
$\overline{{\mathcal H}_c}$ is the closure of ${\mathcal H}_c$, and therefore
${\mathcal H}_a+{\mathcal H}_b\subseteq{\mathcal H}_a\cup{\mathcal H}_b$.
When ${\mathcal H}$ is finite-dimensional or when
the closed subspaces ${\mathcal H}_a$ and  ${\mathcal H}_b$ are orthogonal
to each other then ${\mathcal H}_a+{\mathcal H}_b={\mathcal H}_a\cup{\mathcal H}_b$.
\cite[pp.~21-29]{halmos}, \  \cite[pp.~66,67]{kalmb83}, \
\cite[pp.~8-16]{mittelstaedt-book}

We briefly recall the definitions we will need.  For further
information, see Refs.~\cite{beran,mpoa99,pm-ql-l-hql1,pm-ql-l-hql2}.

\begin{defn}\label{def:lattice}{\rm \cite{birk2nd}}
A {\em lattice} is an algebra
${\rm L}=\langle\mathcal{L}_{\rm O},\cap,\cup\rangle$
such that the following conditions are satisfied for any
$a,b,c\in\mathcal{L}_{\rm O}$:  $a\cup b=b\cup a,\ \  a\cap b=b\cap a,\ \
(a\cup b)\cup c=a,\cup(b\cup c)\ \ (a\cap b)\cap c=a\cap(b\cap c),\ \
a\cap (a\cup b)=a,\ \ a\cup (a\cap b)=a.$
\end{defn}

\begin{thm}\label{th:ordering}{\rm \cite{birk2nd}}
The binary relation $\le$ defined on {\rm L} as
$a\le b\ {\buildrel{\rm def}\over\Longleftrightarrow}
\ a=a\cap b$
is a partial ordering.
\end{thm}

\begin{defn}{\rm \cite{birk3rd}}
An {\em ortholattice} {\rm (OL)} is an algebra
$\langle\mathcal{L}_{\rm O},',\cap,\cup,0,1\rangle$
such that $\langle\mathcal{L}_{\rm O},\cap,\cup\rangle$ is a lattice
with unary operation $'$ called {\em orthocomplementation}
which satisfies the following
conditions for $a,b\in\mathcal{L}_{\rm O}$ ($a'$ is called
the {\em orthocomplement} of $a$):
$a\cup a'=1, \ \ a\cap a'=0,\ \
\le b\ \Rightarrow \ b'\le a',\ \ a''=a.$
\end{defn}

\begin{defn}\label{def:oml-o}{\rm \cite{pav93,p98}}
An {\em orthomodular lattice} {\rm (OML)} is an {\rm OL}
in which the following condition holds:
$a\leftrightarrow b=1\ \Leftrightarrow\ a=b$
where  $a\leftrightarrow b=1 \ {\buildrel{\rm def}\over\Longleftrightarrow}
\ a\to b=1 \ \&\ \ b\to a=1$, where
$a\to b\ {\buildrel\rm def\over =}\ a'\cup(a\cap b)$
\end{defn}

\begin{defn}\label{def:hl}{\rm \protect{\footnote{For additional
definitions of the terms used in this section see
Refs.~\cite{beltr-cass-book,holl95,kalmb86,mpoa99}.}}}
An orthomodular lattice which satisfies the following
con\-di\-tions is a {\em Hilbert lattice}, {\rm HL}.
\begin{enumerate}
\item {\em Completeness:\/}
The meet and join of any subset of
an {\rm HL} exist.
\item {\em Atomicity:\/}
Every non-zero element in an {\rm HL} is greater
than or equal to an atom. (An atom $a$ is a non-zero lattice element
with $0< b\le a$ only if $b=a$.)
\item {\em Superposition principle:\/}
(The atom $c$
is a superposition of the atoms $a$ and $b$ if
$c\ne a$, $c\ne b$, and $c\le a\cup b$.)
\begin{description}
\item[{\rm (a)}] Given two different atoms $a$ and $b$, there is at least
one other atom $c$, $c\ne a$ and $c\ne b$, that is a superposition
of $a$ and $b$.
\item[{\rm (b)}] If the atom $c$ is a superposition of  distinct atoms
$a$ and $b$, then atom $a$ is a superposition of atoms $b$ and $c$.
\end{description}
\item {\em Minimum height:\/} The lattice contains at least
two elements $a,b$ satisfying: $0<a<b<1$.
\end{enumerate}
\end{defn}

Note that atoms correspond to pure states when defined on the lattice.
We recall that {\it irreducibility\/} and the {\it covering
property\/} follow from the superposition principle.\
\cite[pp.~166,167]{beltr-cass-book} We also recall that any Hilbert
lattice must contain a countably infinite number of atoms.
\cite{ivertsj}

Orthogonal vectors determine directions in which we
can orient our detection devices and therefore also
directions of observable projections.
We can choose one-dimensional subspaces ${\mathcal H}_a,\dots
,{\mathcal H}_e$ as shown in Fig.~\ref{fig:mmp}, where we
denote them as $a,\dots,e$. The Hasse lattice shown in
the figure graphically represents the orthogonality
between the vectors---in our case the ones between
each chosen vector and a plane determined by the other two.
In particular, the orthogonalities are $a\perp b,c,d,e$ since
$a\le b',c',d',e'$, $b\perp c$ since $a\le c'$, and
$d\perp e$ since $d\le e'$. Also, e.g., $b'$ is a
complement of $b$ and that means a plane to which $b$
is orthogonal: $b'=a\cup c$. Eventually $b\cup b'=1$ where
$1$ stands for $\mathcal H$. Greechie lattices are shorthand
representations of a certain class of Hasse lattices. The one
corresponding to our Hasse lattice above is shown in
Fig.~\ref{fig:mmp}.

\begin{figure}[htp]\centering
  \setlength{\unitlength}{0.8pt}
  \begin{picture}(330,100)(0,0)
    \put(10,0){
      \begin{picture}(60,80)(0,0)
        \qbezier(0,40)(30,60)(60,40)
        \put(30,25){\line(0,1){50}}
        \put(30,75){\circle*{7}}
        \put(30,25){\circle*{7}}
        \put(30,50){\circle*{7}}
        \put(0,40){\circle*{7}}
        \put(60,40){\circle*{7}}
        \put(0,50){\makebox(0,0)[b]{$d$}}
        \put(40,54){\makebox(0,0)[b]{$a$}}
        \put(60,50){\makebox(0,0)[b]{$e$}}
        \put(40,25){\makebox(0,0)[l]{$b$}}
        \put(40,75){\makebox(0,0)[l]{$c$}}
      \end{picture}
    } 
    \put(150,5){
      \begin{picture}(60,90)(0,0)
        \put(60,90){\line(-2,-1){60}}
        \put(60,90){\line(-1,-1){30}}
        \put(60,90){\line(0,-1){30}}
        \put(60,90){\line(1,-1){30}}
        \put(60,90){\line(2,-1){60}}
        \put(60,0){\line(-2,1){60}}
        \put(60,0){\line(-1,1){30}}
        \put(60,0){\line(0,1){30}}
        \put(60,0){\line(1,1){30}}
        \put(60,0){\line(2,1){60}}
        \put(0,60){\line(1,-1){30}}
        \put(0,60){\line(2,-1){60}}
        \put(30,60){\line(-1,-1){30}}
        \put(30,60){\line(1,-1){30}}
        \put(60,60){\line(-2,-1){60}}
        \put(60,60){\line(-1,-1){30}}
        \put(60,60){\line(1,-1){30}}
        \put(60,60){\line(2,-1){60}}
        \put(90,60){\line(-1,-1){30}}
        \put(90,60){\line(1,-1){30}}
        \put(120,60){\line(-2,-1){60}}
        \put(120,60){\line(-1,-1){30}}

        \put(60,-5){\makebox(0,0)[t]{$0$}}
        \put(-5,30){\makebox(0,0)[r]{$d$}}
        \put(18,28){\makebox(0,0)[l]{$e$}}
        \put(65,28){\makebox(0,0)[l]{$a$}}
        \put(96,30){\makebox(0,0)[l]{$b$}}
        \put(125,30){\makebox(0,0)[l]{$c$}}
        \put(-5,60){\makebox(0,0)[r]{$d'$}}
        \put(17,62){\makebox(0,0)[l]{$e'$}}
        \put(65,65){\makebox(0,0)[l]{$a'$}}
        \put(96,62){\makebox(0,0)[l]{$b'$}}
        \put(125,60){\makebox(0,0)[l]{$c'$}}
        \put(60,95){\makebox(0,0)[b]{$1$}}
        \put(60,0){\circle*{7}}
        \put(0,30){\circle*{7}}
        \put(30,30){\circle*{7}}
        \put(60,30){\circle*{7}}
        \put(90,30){\circle*{7}}
        \put(120,30){\circle*{7}}
        \put(0,60){\circle*{7}}
        \put(30,60){\circle*{7}}
        \put(60,60){\circle*{7}}
        \put(90,60){\circle*{7}}
        \put(120,60){\circle*{7}}
        \put(60,90){\circle*{7}}
      \end{picture}
    } 
  \end{picture}
  \caption{MMP/Greechie lattice, and its corresponding Hasse lattice.
\label{fig:mmp}}
\end{figure}
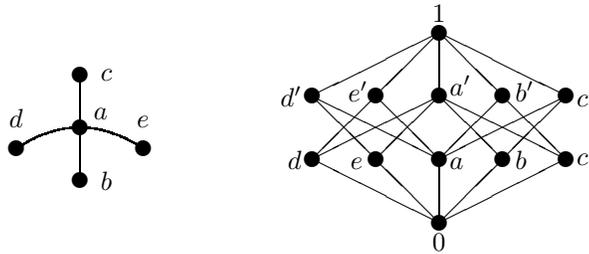

The Hasse lattice shown in Fig.\ \ref{fig:mmp} is a subalgebra of
a Hilbert lattice but, as we show below, already the one with
a third orthogonal triple attached to it is not.
Therefore, for generation of our lattices we should instead
use MMP hypergraphs to which we shall
ascribe a lattice meaning later on.
We define MMP hypergraphs (also called MMP diagrams) as
follows\ \cite{pmmm04b}
\begin{itemize}
\item[(i)] Every vertex belongs to at least one edge;
\item[(ii)] Every edge contains at least 3 vertices;
\item[(iii)] Edges that intersect each other in $n-2$
         vertices contain at least $n$ vertices.
\end{itemize}

\font\1=cmss8
\font\2=cmssdc8
\font\3=cmr8

We encode MMP hypergraphs by means of alphanumeric
and other printable ASCII characters. Each vertex (atom)
is represented by one of the following
characters: {\11\hfil $\:$2\hfil $\:$3\hfil $\:$4\hfil
$\:$5\hfil $\:$6\hfil $\:$7\hfil $\:$8\hfil $\:$9\hfil
$\:$A\hfil $\:$B\hfil $\:$C\hfil $\:$D\hfil $\:$E\hfil
$\:$F\hfil $\:$G\hfil $\:$H\hfil $\:$I\hfil $\:$J\hfil
$\:$K\hfil $\:$L\hfil $\:$M\hfil $\:$N\hfil $\:$O\hfil
$\:$P\hfil $\:$Q\hfil $\:$R\hfil $\:$S\hfil $\:$T\hfil
$\:$U\hfil $\:$V\hfil $\:$W\hfil $\:$X\hfil $\:$Y\hfil
$\:$Z\hfil $\:$a\hfil $\:$b\hfil $\:$c\hfil $\:$d\hfil
$\:$e\hfil $\:$f\hfil $\:$g\hfil $\:$h\hfil $\:$i\hfil
$\:$j\hfil $\:$k\hfil $\:$l\hfil $\:$m\hfil $\:$n\hfil
$\:$o\hfil $\:$p\hfil $\:$q\hfil $\:$r\hfil $\:$s\hfil
$\:$t\hfil $\:$u\hfil $\:$v\hfil $\:$w\hfil $\:$x\hfil
$\:$y\hfil $\:$z\hfil $\:$!\hfil $\:$"\hfil $\:$\#\hfil
$\:${\scriptsize\$}\hfil $\:$\%\hfil $\:$\&\hfil $\:$'\hfil $\:$(\hfil
$\:$)\hfil $\:$*\hfil $\:$-\hfil $\:$/\hfil $\:$:\hfil
$\:$;\hfil $\:$$<$\hfil $\:$=\hfil $\:$$>$\hfil $\:$?\hfil
$\:$@\hfil $\:$[\hfil $\:${\scriptsize$\backslash$}\hfil $\:$]\hfil
$\:$\^{}\hfil $\:$\_\hfil $\:${\scriptsize$\grave{}$}\hfil
$\:${\scriptsize\{}\hfil
$\:${\scriptsize$|$}\hfil $\:${\scriptsize\}}\hfil
$\:${\scriptsize\~{}}}\ , and then again all these characters prefixed
by `+', then prefixed by `++', etc. There is no upper limit on the
number of characters.

Each block is represented by a string of characters that
represent atoms (without spaces). Blocks are separated by
commas (without spaces). All blocks in a line form a
representation of a hypergraph.  The order of the blocks is
irrelevant---however, we shall often present them
starting with blocks forming the biggest loop to facilitate
their possible drawing.  The line must end with a full stop
(i.e.\ a period).
Skipping of characters is allowed.

{\em Generalized orthoarguesian equations} $n${\rm OA}\
\cite{mpoa99,pm-ql-l-hql2} that hold in any Hilbert lattice
follow from the following set of equations that hold in any
Hilbert space.

\begin{thm} \label{th:hs-ssnoa}
Let ${\mathcal M}_0,\ldots, {\mathcal M}_n$ and ${\mathcal N}_0,\ldots,
{\mathcal N}_n$, $n\ge 1$,
 be any subspaces (not necessarily closed) of a Hilbert
space, and let $\bigcap$ denote set-theoretical
intersection and $+$ subspace sum.
We define the subspace term ${\mathcal T}_n(i_0,\ldots,i_n)$ recursively as
follows, where $0\le i_0,\ldots,i_n\le n$:
\begin{align}
&{\mathcal T}_1(i_0,i_1)=({\mathcal M}_{i_0}
         +{\mathcal M}_{i_1})\text{$\bigcap$}({\mathcal N}_{i_0}
         +{\mathcal N}_{i_1}) \label{eq:hs-rec1} \\
{\mathcal T}_m&(i_0,\ldots,i_m)={\mathcal T}_{m-1}(i_0,i_1,i_3,\ldots,i_m)
          \notag \\
&\text{$\bigcap$}({\mathcal T}_{m-1}
(i_0,i_2,i_3,\ldots,i_m)+{\mathcal  T}_{n-1}(i_1,i_2,i_3,\ldots,i_m)),\quad
2\le m\le n    \label{eq:hs-recn}
\end{align}

For $m=2$, this means ${\mathcal T}_2(i_0,i_1,i_2)={\mathcal T}_1(i_0,i_1)$
$\bigcap\ ({\mathcal T}_1(i_0,i_2)+{\mathcal T}_1(i_1,i_2))$.  Then the following condition holds in any
finite- or infinite-dimensional Hilbert space for $n\ge 1$:
\begin{align}
({\mathcal M}_0+{\mathcal N}_0)\text{$\bigcap$}
\cdots\text{$\bigcap$}({\mathcal M}_n+{\mathcal N}_n)\subseteq {\mathcal N}_0
+({\mathcal M}_0\text{$\bigcap$}
         ({\mathcal M}_1+{\mathcal T}_n(0,\ldots,n) )).  \label{eq:hs-ssnoa}
\end{align}
\end{thm}

\begin{proof} As given in \cite{mpoa99-arXiv,bdm-ndm-mp-fresl-10}
\end{proof}

We will use the above theorem to derive a condition that holds in the
lattice of closed subspaces of a Hilbert space. In doing so we will
make use of the definitions introduced above and  the following
well-known \cite[p.\ 28]{halmos} lemma.
\begin{lem}
\label{lem:hs-sum} Let ${\mathcal M}$ and ${\mathcal N}$ be two closed
subspaces of a Hilbert space. Then
\begin{align}
{\mathcal M}+{\mathcal N} & \subseteq {\mathcal M}\text{$\bigcup$} {\mathcal N} \label{eq:hs-sumss} \\
{\mathcal M}\perp {\mathcal N} \quad & \Rightarrow\quad {\mathcal M}+{\mathcal N}= {\mathcal M}\text{$\bigcup$} {\mathcal N}
\label{eq:hs-sumeq}
\end{align}
\end{lem}

\begin{thm} \label{th:hs-noa} {\rm (Generalized Orthoarguesian Laws)}
Let ${\mathcal M}_0,\ldots, {\mathcal M}_n$ and ${\mathcal N}_0,\ldots,
{\mathcal N}_n$, $n\ge 1$, be closed subspaces of a Hilbert space.
We define the term ${\mathcal T}^{\mbox{\tiny{$\bigcup$}}}_n(i_0,\ldots,i_n)$ by
substituting $\bigcup$ for $+$ in the term
${\mathcal T}_n(i_0,\ldots,i_n)$ from Theorem~\ref{th:hs-ssnoa}.
Then following condition holds in any finite- or
infinite-dimensional Hilbert space for
$n\ge 1$:
\begin{align}
{\mathcal M}_0 & \perp {\mathcal N}_0 \ \& \ \cdots \ \& \
      {\mathcal M}_n\perp {\mathcal N}_n \quad\Rightarrow\quad
          \notag \\
  &  ({\mathcal M}_0\text{$\bigcup$} {\mathcal N}_0)\text{$\bigcap$}
\cdots\text{$\bigcap$}({\mathcal M}_n\text{$\bigcup$} {\mathcal N}_n)
 \le
          {\mathcal N}_0\text{$\bigcup$} ({\mathcal M}_0\text{$\bigcap$}
         ({\mathcal M}_1\text{$\bigcup$}
              {\mathcal T}^\text{$\bigcup$}_n(0,\ldots,n) )). \label{eq:hs-noa}
\end{align}
\end{thm}
\begin{proof} By the orthogonality hypotheses and Eq.~(\ref{eq:hs-sumeq}), the
left-hand side of Eq.~(\ref{eq:hs-noa}) equals the left-hand side
of Eq.~(\ref{eq:hs-ssnoa}).  By Eq.~(\ref{eq:hs-sumss}),
the right-hand side of Eq.~(\ref{eq:hs-ssnoa}) is a subset
of the right-hand side of Eq.~(\ref{eq:hs-noa}).  Eq.~(\ref{eq:hs-noa})
follows by Theorem~\ref{th:hs-ssnoa} and the transitivity of
the subset relation.
\end{proof}

Ref.~\cite{mpoa99} shows that in any OML (which includes the lattice of
closed subspaces of a Hilbert space, i.e., the Hilbert lattice),
Eq.~(\ref{eq:hs-noa}) is equivalent to the $m$OA law Eq.~(\ref{eq:noa})
for $m=n+2$, thus establishing the proof of Theorem~\ref{th:noa}.

\begin{defn}
\label{def:noa}
We define an operation
${\buildrel (n)\over\equiv}$ on $n$ variables
$a_1,\ldots,a_n$ ($n\ge 3$) as follows:
\begin{align}
a_1&\ {\buildrel (3)\over\equiv}a_2\
{\buildrel\rm def\over =}\
((a_1\to  a_3)\cap(a_2\to  a_3))
\cup((a_1'\to  a_3)\cap(a_2'\to  a_3)) \notag\\
a_1&{\buildrel (n)\over\equiv}a_2\
{\buildrel\rm def\over =}\ (a_1{\buildrel (n-1)\over\equiv}a_2)\cup
((a_1{\buildrel (n-1)\over\equiv}a_n)\cap
(a_2{\buildrel (n-1)\over\equiv}a_n)), \qquad n\ge 4\,.\label{eq:noaoper}
\end{align}
\end{defn}

\begin{thm}\label{th:noa}
The $n${\rm OA} {\em laws}
\begin{align}
(a_1\to a_3) \cap (a_1{\buildrel (n)\over\equiv}a_2)
\le a_2\to  a_3\,.\label{eq:noa}
\end{align}
hold in any Hilbert lattice.
\end{thm}

The class of  equations (\ref{eq:noa}) are the {\em generalized
orthoarguesian equations} $n${\rm OA}. \cite{mpoa99,pm-ql-l-hql2}

\section{\label{sec:represent}Lattices That Describe Kochen-Specker Sets}

The Kochen-Specker (KS) theorem claims that
experimental recordings that cannot be predetermined, i.e.\
fixed in advance. Its best known proof is based on sets (KS sets)
to which it was impossible to ascribe classical 0-1 values.
Two such sets are shown in Figs.\ \ref{fig:bub} and \ref{fig:peres}.

\begin{figure}[htp]
\begin{center}
\includegraphics[width=0.98\textwidth,height=0.21\textwidth]{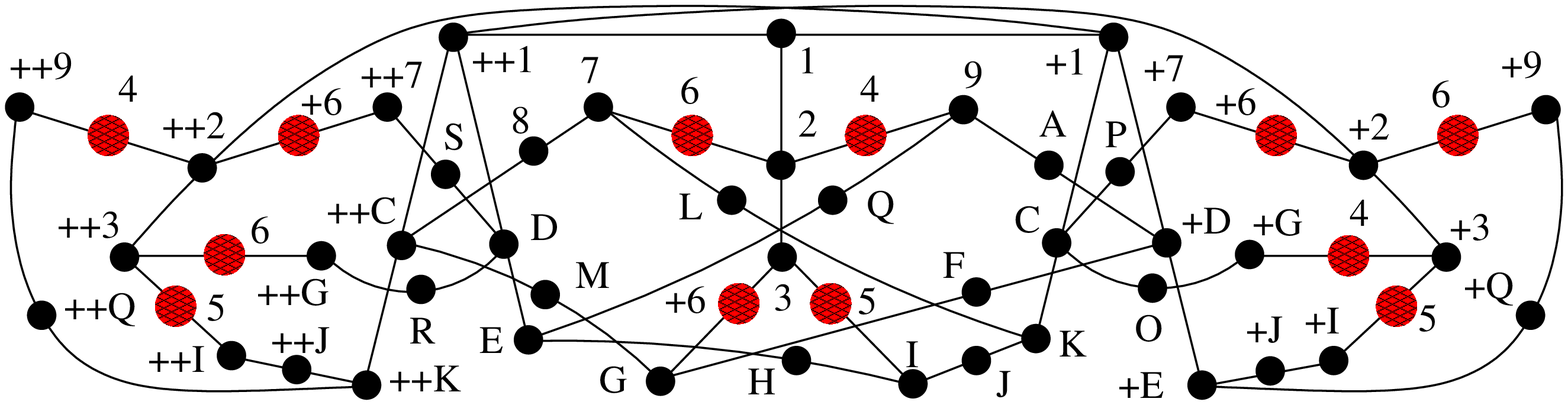}
\end{center}
\caption{Bub's MMP with 49 atoms and 36 blocks. Notice that 12 bigger
dots with a pattern (red online) represent just 4 atoms: 4, 5, 6, and +6.}
\label{fig:bub}
\end{figure}

Bub's set, shown in Fig.\ \ref{fig:bub} is the smallest known KS setup.\cite{bub} Its MMP hypergraph reads:
{\1123,  249,  267,  78++C,  9A+D,  +1CK,  ++1DE,  7LK,  9QE,  35I,  3+6G,  EHI,  IJK,  +DFG,  GM++C,  CP+7,  CO+G,  ++1++C++K,  ++3++2+1,  +6++7++2,  ++24++9,  ++9++Q++K,  ++35++I,  ++I++J++K,  ++36++G,  ++GRD,  DS++7,  +3+2++1,  +7+6+2,  +26+9,  +34+G,  +35+I,  +I+J+E,  +1+D+E,  +E+Q+9,  1+1++1.} It is shown in Fig.~\ref{fig:bub}

Peres' set \cite{peres}, shown in Fig.\ \ref{fig:peres}  is the most symmetric KS set among
those with less than 60 vectors---it has 57
 vectors (vertices) and 40 tetrads (edges).  Its MMP hypergraph reads:
{\1123,  345,  467,  789,  92A,  ABC,  CD4,  AE+J,  5F+J,  IG+9,  IH+5,  I7+1,  JC++1,  ++1+2+3,  +3+4+5,  +4+6+7,  +7+8+9,  +9+2+A,  +A+B+C,  +C+D+4,  +A+E++J,  +5+F++J,  +I+G++9,  +I+71,  +I+H++5,  +J+C+1,  +1++2++3,  ++3++4++5,  ++4++6++7,  ++7++8++9,  ++9++2++A, ++A++B++C,  ++C++D++4,  ++A++EJ,\break ++5++FJ,  ++I++G9, ++I++7++1,  ++I++H5,  ++J++C1,  1+1++1.} Another highly symmetrical KS set is the original Kochen-Specker's one \cite{koch-speck}
but it contains 192 vectors.

\begin{figure}[htp]
\begin{center}
\includegraphics[width=0.98\textwidth,height=0.25\textwidth]{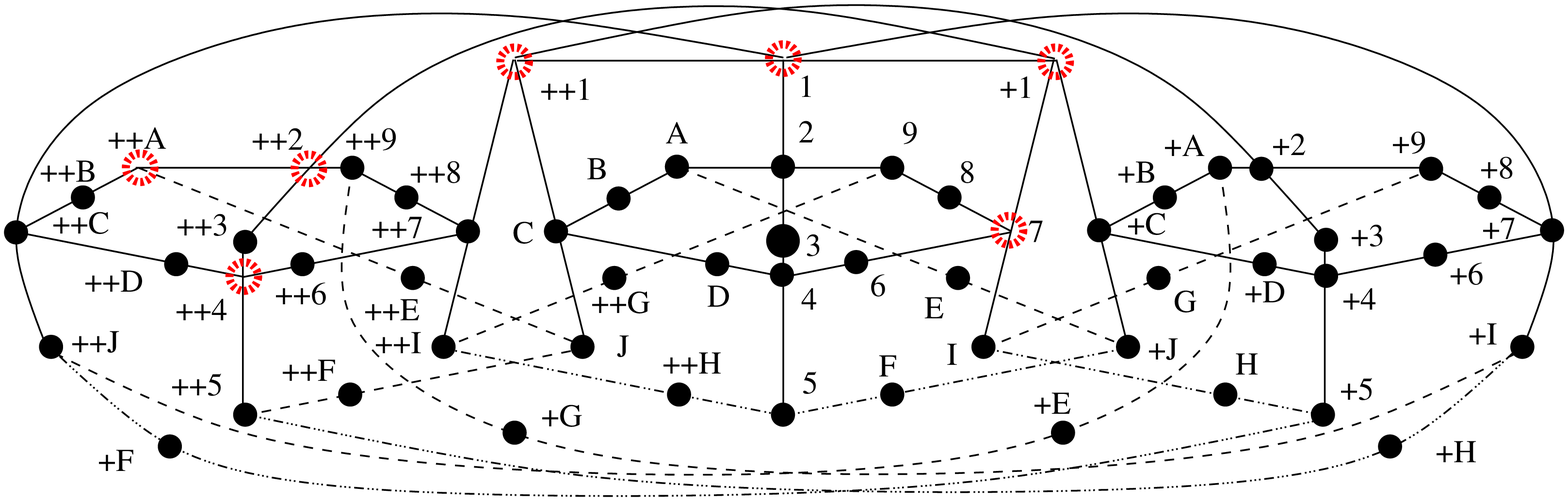}
\end{center}
\caption{Peres' KS MMP hypergraph/lattice. Red (online) rings denote atoms at
which Peres' lattice violates 7OA i.e. the failing assignment of
atoms or co-atoms to the variables of 7OA in the form of
Eq.~(\ref{eq:noa}).}
\label{fig:peres}
\end{figure}

Now, a number of authors have represented KS setups or indeed
any spin-1 experimental setup by means of Greechie
lattices.\cite{shimony,hultgren,svozil-tkadlec,svozil-book-ql,n:tka98,tkadlec,tkadlec-2,smith03,foulis99}

As we show above, the Hilbert lattice of any quantum system
has to satisfy $n$OA equations.
If we assume that the hypergraphs that describe Peres' and Bub's setups
can be represented
by lattices, we would end up with Greechie lattices for them,
i.e., lattices that recognize
only relations between orthogonal atoms and coatoms (spans) from
such orthogonal
sets. When we check---by our program {\tt latticeg} described in
Sec.\ \ref{sec:algo}---whether the Greechie lattices pass
$n$OA equations, we find out that Bub's lattice violates 3OA
(and of course all
$n$OA, $n > 3$) and that Peres' satisfies 3OA-6OA and violates 7OA.
The reason that happens is simple: Greechie diagrams are
not subalgebras of a
Hilbert lattice and the aforementioned authors apparently
did not realize that.

To convince ourselves that Peres' and Bub's Hilbert lattices
really do satisfy 7OA and
8OA, it is enough to invoke Th.\ \ref{th:hs-noa} according
to which any quantum
system (set of vectors/states ascribed to it) has to satisfy
all $n$OA equations.
But let us nevertheless go into some details with Bub's
Hilbert vectors so as to
arrive at proper lattices and proper Hasse diagrams
that they have to use.  A proper
description can only be carried out with lattices and
Hasse diagrams that take into account
joins (spans in terms of vectors) of nonorthogonal atoms
(vectors) as well as the joins and meets (spans
and intersections, respectively) of those joins, etc.

The details are as follows. We consider Bub's KS setup. To be able
to apply our program {\tt vectorfind} for finding the vector
components of Bub's setup shown in Fig.\ \ref{fig:bub}, we have to
write down its MMP representation without gaps in letters. So, we
have {\1123,\dots,DFH,\dots}, where we present only those
Greechie/Hasse lattice atoms in which 3OA failed. Their
Hilbert space vectors are: {\11=\{0,0,1\}}, {\12=\{1,0,0\}},
{\1F=\{1,-2,-1\}}, and {\1D=\{1,1,-1\}}.

In a Hilbert space representation, Bub's KS setup does pass 3OA.
Let us consider 3OA in the following form \
\begin{align}
&a\perp b  \quad\&\quad q\perp n\quad\notag\\
       & \Rightarrow
(a\cup b)\cap (q\cup n)\le b\cup(a\cap
(q\cup((a\cup q)
\cap (b\cup n)))).\notag
\end{align}
In 3-dim Euclidean space, all subspaces are closed (they are lines,
planes, or the whole space), so $a\cup b=a+b$, i.e., subspace join
and subspace sum are the same. Thus, converting joins in the previous
equation to subspace sums and using the orthogonality we get:
\begin{align}
a\perp b & \quad\&\quad q\perp n
\Rightarrow (a+b)\cap(q+n)\notag\\
&\le b+(a\cap(q+((a+q)\cap (b+n)))).\label{eq:bubs-proof}
\end{align}

Now, using the subspaces determined by the aforementioned vectors
and their spans in a Hilbert space, we can easily check that
Bub's representation pass 3OA. For instance, vectors {\11},
{\12}, {\1F}, and {\1D}, determine subspaces {\1\{0,0,$\alpha$\}},
{\1\{$\beta$,0,0\}}, {\1\{$\gamma$,-2$\gamma$,-$\gamma$\}}, and
{\1\{$\delta$,$\delta$,-$\delta$\}}, with arbitrary coefficients
$\alpha,\dots$. They represent lines in both 3-dim Hilbert
space and 3-dim Euclidean space. {\1\{0,0,$\alpha$\}}+{\1\{$\beta$,0,0\}}=
{\1\{$\beta$,0,$\alpha$\}} is a plane spanned by {\11} and {\12}, etc.
We show a verification of Eq.\ (\ref{eq:bubs-proof}) in
Fig.\ \ref{fig:bub-proof}.

\begin{figure}[htp]
\begin{center}
\includegraphics[width=0.47\textwidth,height=0.4\textwidth]{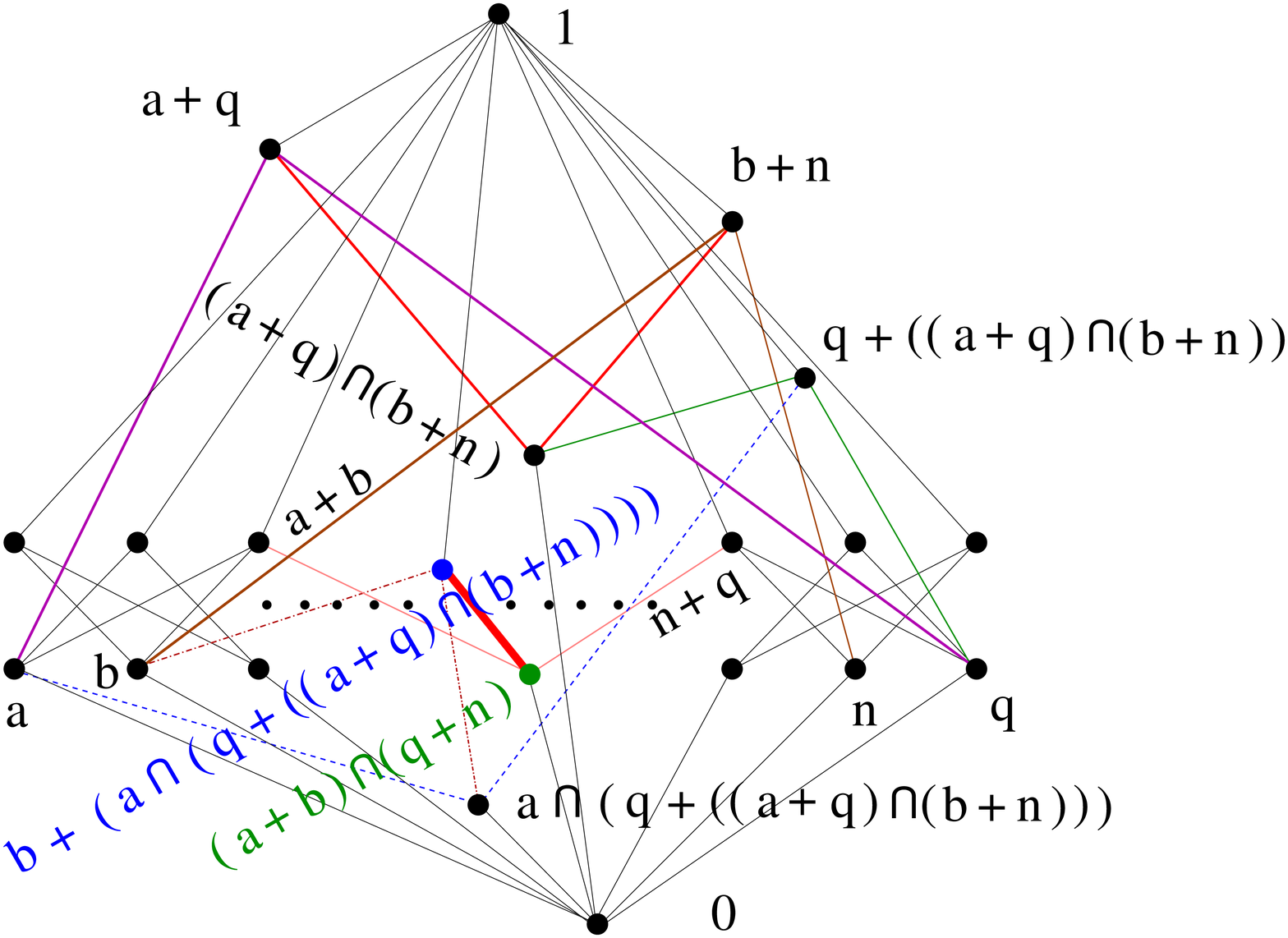}
\end{center}
\caption{A new kind of lattice (MMPL) in which Bub's setup passes 3OA.
The inequality relation in Eq.\ (\ref{eq:bubs-proof}) is represented
by the thick line (red online).}
\label{fig:bub-proof}
\end{figure}

Such lattices---we call them MMPLs---are essential for checking various
other equations, because, e.g., the MMPL shown in Fig.\ \ref{fig:bub-proof}
as an example of a lattice that satisfies 3OA for a particular nodal assignment
to its variables, and we can further check whether
it satisfies other equations that correspond to a particular experimental
setup. Thus when we need a lattice to set up a blueprint for an experiment
in which it is important that a system satisfy particular equations, we shall
use MMPL. When we just need to find  a lattice in which an equation fails
and another pass to show their independence, a Greechie lattice might
serve us better.  Greechie lattices contain only relations between elements
within orthogonal subsets of chosen lattices and therefore for more
complicated equations soon become so large that one cannot
compute them any more.  Thus we were actually lucky to find
that Peres' lattice satisfied 6OA and violated 7OA because that
provided an immediate proof that 7OA does not follow from 6OA.

\section{Main Result: Lattices That Satisfy 6OA and Violate 7OA}
\label{sec:main}

Peres' lattice violates 7OA at  {\1++1, ++4, 1, 7, +1, ++A, ++23},
and we have indicated these elements with the help of rings in
Fig.\ \ref{fig:peres}. But rather than analyze the failure, we will
show how we can arrive at a much smaller lattice that also satisfies
6OA and violates 7OA.  The procedure shows how we can get
smaller lattices using our program {\tt latticeg} to eliminate atoms and
blocks that did not take part in the violations of 7OA we originally
found.

When we apply {\tt latticeg} to the equation 7OA and it arrives at atoms
(or more precisely, lattice nodes) at which 7AO fails, the program gives
the nodes we listed above, and it also gives us the
following additional information about the failure:

{\noindent\tt Greechie atoms not visited:  2 3 4  \dots}

{\noindent\tt Greechie blocks that don't affect the failure: 345 ABC CD4 \dots}

If, during the evaluation of the failing assignment, the meets and
joins contained in a block are never used, then that block is unrelated
to the failure.   The program accumulates such
blocks and puts them into
a list called ``don't affect the failure'' as illustrated by the
sample printout above.  After removing these from the Peres' Greechie
lattice of Fig.\ \ref{fig:peres} and renaming the atoms, we end up
with the smaller Greechie
lattice {\1123, 345, 567, 789, 9AB, BCD, DEF, FGH, HIJ, JKL, LMN, NOP, PQR,
RS1, 4EK, 4AP, AVH, BXL, DUQ, FWN, JTQ} \ \ which is shown in Fig.\
\ref{fig:7oa01}. The left figure shows the blocks we dropped from
Fig.\ \ref{fig:peres}, and the right one is given in the representation
we previously used to show violations of 3OA through 6OA at
lattices presented in \cite{mpoa99-arXiv,pm-ql-l-hql2,mp-alg-hilb-eq-09}
with the maximal loop (tetrakaidecagon, 14-gon) it contains.

\begin{figure}[htp]
\begin{center}
\includegraphics[width=0.65\textwidth,height=0.25\textwidth]{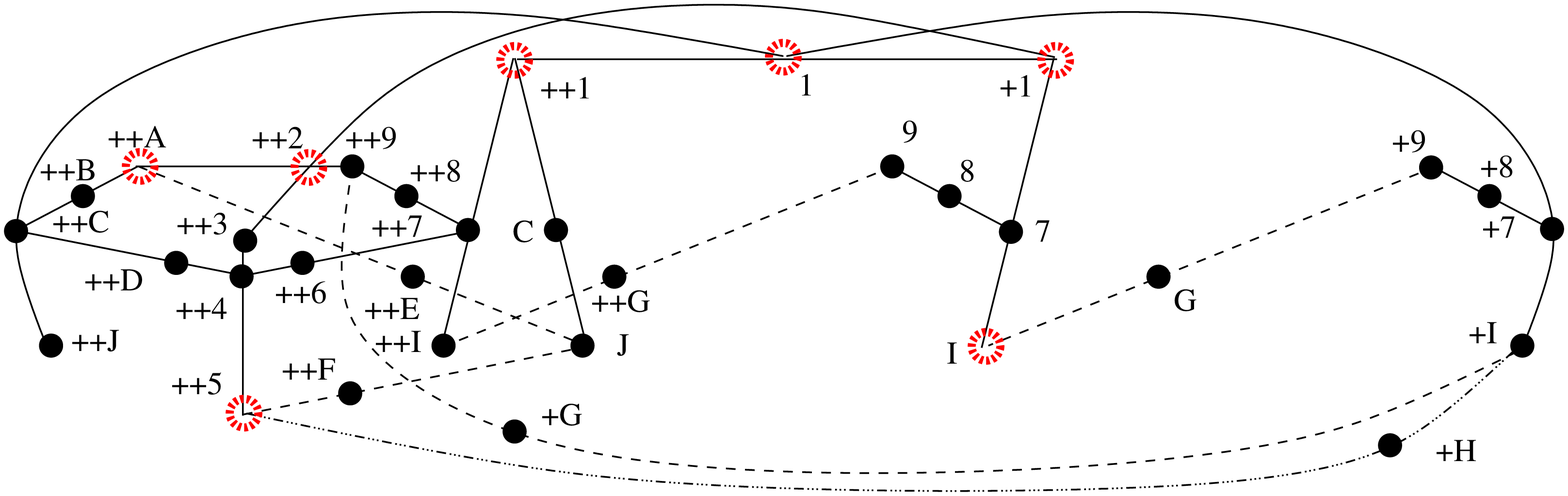}
\includegraphics[width=0.34\textwidth]{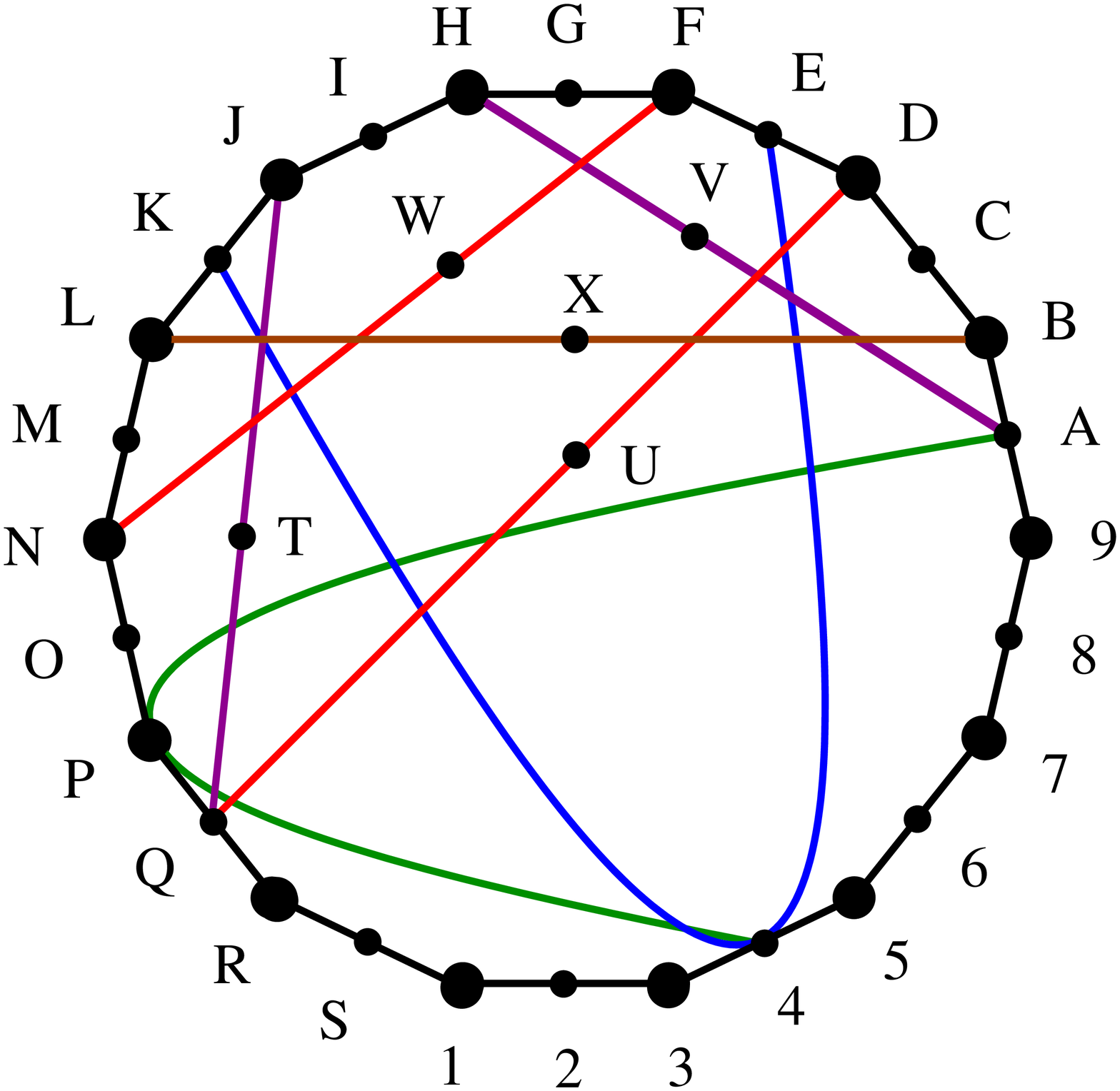}
\end{center}
\caption{A lattice with 33 atoms and 21 blocks that satisfies 6OA
and violates 7OA.
Red (online) rings show atoms that take part in a violation of 7OA. The left
and right diagrams are isomorphic to each other (i.e.\ are two
ways of drawing the same lattice).}
\label{fig:7oa01}
\end{figure}

Of course, there is a whole series of lattices between Peres' 57-40
and the 33-21 shown here with the same property of violating 7OA
and satisfying 6OA which we obtain by adding the removed blocks to
33-21 lattice until we obtain Peres'.

The independence of 7OA emerged from our study to determine which
quantum properties continue to hold when 3-dim KS setups
are approximately (but erroneously, in a strict mathematical
sense) represented by Greechie lattices.  The passing of 6OA
and failing of 7OA was fortuitous and quite unexpected.
Previously, we had little hope of finding such an example with
the help of Greechie diagrams. The discovery of a lattice passing 5OA
but failing 6OA required many weeks on a 500-CPU cluster, and that
discovery itself involved a large element of luck combined with some
judicious intuition by the second author about which lattices
might be promising.  The search for a 7OA counterexample was
expected to be many orders of magnitude harder.  Even the
verification that the single Peres' lattice passed 6OA
required weeks of cluster time, and had it not been for
an early occurrence of a failure in the 7OA test, that test
might have required a much longer time.

\section{Algorithms and Programs}
\label{sec:algo}

The main program that we used for this work was
{\tt latticeg}, which is a general-purpose utility for
testing equations against orthocomplemented lattices expressed
in the form of Greechie diagrams.  Its algorithm is
described in Ref.~\cite{bdm-ndm-mp-1-arXiv}.

The $n$OA law in the form derived directly from Hilbert space,
Eq.~(\ref{eq:hs-noa}), has $2n-2$ variables, whereas in the
equivalent form of Eq.~(\ref{eq:noa}) it has $n$ variables.
Since testing an equation with $m$ variables against a lattice
with $k$ nodes requires that up to $k^m$ combinations be checked,
it is more efficient to use the form of Eq.~(\ref{eq:noa}).

Eq.~(\ref{eq:noa}) has $8\cdot 3^{n-3}+4$ occurrences of its $n$
variables.  For faster computation, we found an equivalent with $6\cdot
3^{n-3}+3$ variable occurrences (which equals 166 for 6OA and 489 for 7OA).
The following theorem shows this equivalent form for $n=3$.  The proof
is similar for larger $n$.  The general form for larger $n$ can be
inferred by looking at the proof, although we have not defined a
``compact'' notation for it as we have for Eq.~(\ref{eq:noa}).
\begin{thm}\label{th:3oaa}
An {\rm OML} in which the equation
\begin{align}
a\cap ((a\cap b)\cup ((a\to c)\cap (b\to c)))\le b'\to c  \label{eq:3oaa}
\end{align}
holds is a {\rm 3OA} and vice-versa.
\end{thm}
\begin{proof}   {\bf For Eq.~(\ref{eq:3oaa}):}
To obtain the 3OA law,
Eq.~(\ref{eq:noa}), from Eq.~(\ref{eq:3oaa}), we substitute
$a\to c$ for $a$ and $b\to c$ for $b$, then we use the OML identities
$(a\to c)\to c=a'\to c$,
$(b\to c)\to c=b'\to c$, and
$(b'\to c)\to c=b\to c$.

For the converse, since $x\le x'\to y$,
\begin{align}
a\cap ((&a\cap b)\cup ((a\to c)\cap (b\to c))) \notag \\
&\le(a'\to c)\cap (((a'\to c)\cap (b'\to c))\cup ((a\to c)\cap (b\to c)))
         \notag \\
&=(a'\to c) \cap (a'{\buildrel c\over\equiv}b') \notag \\
&\le b'\to c,         \notag
\end{align}
where for the last step we used an instance of
Eq.~(\ref{eq:noa}) for $n=3$.
\end{proof}

Because of the large size of the $n$OA equations for larger $n$, in
order to ensure that our input to {\tt latticeg} was free from typos
we used an auxiliary utility program, {\tt oagen}, to generate $n$OA
equations in the form of either Eq.~(\ref{eq:noa}) or
Eq.~(\ref{eq:3oaa}).

The evaluation of the 7OA equation on the Peres Greechie diagram
involves 7 nested loops, each with 116 iterations (since its Hasse
diagram has 116 nodes).  For the shorter equation of the form of
Eq.~(\ref{eq:3oaa}), each evaluation at the innermost loop involves an
assignment to 489 variable occurrences and 487 join, meet, and $\to $
operations (the last having a precomputed table in memory from its
join, meet, and orthocomplementation expansion).  Thus $116^7\cdot
489=138,202,145,015,414,784$ (138 quadrillion) operation evaluations
($489 = 487+1+1$ includes the final $\le$ comparison and a single
orthocomplementation) are required for a full scan.

Such a direct, full evaluation is a challenge on today's hardware, even
with a cluster of processors, unless one is very lucky to encounter a
failure early on in the scan (and we were).  In addition, we made
several enhancements to {\tt latticeg} to help make this project more
feasible:
\begin{itemize}

\item The main algorithm was improved.  The original algorithm assigned
each possible combination of lattice nodes to the equation variables,
then evaluated the resulting equation according to the structure of the
lattice (i.e. the suprema, infima, and orthocomplements in the Hasse
diagram derived from the input Greechie diagram).  The main scan
consists of nested loops that processes all nodal assignments to the
first variable in the outermost loop, then all assignments to the second
variable in the next inner loop, and so on.  Since it has 7 variables,
testing the 7OA equation involves 7 nested loops.

The new algorithm takes into account, at each loop level, the variables
in outer loops (which have known assignments) and evaluates as much of
the equation as it can with those known variables.  The equation is then
shrunk with these partial evaluations, for further processing at that
and deeper loop levels.  Eventually, the equation is shrunk to a length
of one, which means that it is completely evaluated.  While a length of
one will always be obtained at the innermost loop level, it may also occur
at an outer level (such as when an expression containing
not-yet-assigned variables is conjoined with a partial evaluation that
resulted in lattice 0).  In such cases, processing of further inner
loops becomes unnecessary.  So, the new algorithm benefits from (1)
shorter equations to evaluate at deeper loop levels and (2) possible
skipping of the deepest loops.  Overall, this results in a speedup of
about a factor of 10 for the 7OA equation evaluation.

Because of the complexity of the new partial evaluation algorithm, it was
put into a new version of {\tt latticeg} called {\tt lattice2g}.
This allows us to check that the old and new algorithms produce the same
result, helping to make sure there isn't a program bug in the new
algorithm.  Having two programs also allows us to directly measure
the speedup afforded by the new algorithm.

\item For testing a huge lattice, a feature was added to break up
the testing into several independent parts.  This way the different parts
can be run on different processors in our cluster.  The test can be
partitioned into any number of outermost and first inner loop
iterations.  For example, the Peres Greechie diagram has a Hasse
representation with 116 nodes.  We can specify that the cluster test the
98th iteration (out of 116) of the outmost loop and the 101st through
110th iteration (out of 116) of the next inner loop.

\item A feature was added to analyze an equation failure to determine
what nodes, atoms, and blocks were not involved in the failure.  In
particular, a block is said not to affect the failure whenever all
operations that ``visit'' (non-0 and non-1) nodes in the block do not
involve any other (non-0 and non-1) nodes in that block.  This is
described in more detail in Sec.~\ref{sec:main}, where we show how this
feature was used to determine which blocks could be removed from Peres'
Greechie lattice to obtain a smaller lattice that satisfies 6OA but
violates 7OA

\end{itemize}

\section{Conclusion}
\label{sec:concl}

After 75 years of research carried out in the field of the algebraic
structure underlying quantum Hilbert space---the Hilbert
lattice---only one class of equations (beyond the orthomodular
lattice laws) that hold in it was
found: the class of orthoarguesian equations. Individual orthoarguesian
equations were found in the eighties and nineties.  All other equations
known to hold in a Hilbert lattice require a state introduced to it.

Then in 2000 we found  \cite{mpoa99} a class ($n$oa) of lattices
determined by {\em generalized orthoarguesian equations} ($n$OA)
and proved that the following
inclusion holds: $n {\rm oa}\subseteq(n+1){\rm oa}$. We also proved
that all previously found OAs are equivalent to either 3OA or 4OA,
we proved that 4OA is strictly stronger than 3OA, and we
found lattices in which 4OA passed but 5OA failed and lattices in which
5OA passed and 6OA failed.\ \cite{pm-ql-l-hql2}

In this paper we found a series of lattices---shown in
Figs.\ \ref{fig:peres} and
\ref{fig:7oa01} and obtained as explained in Sec.\ \ref{sec:main}---in
which 6OA passes and 7OA fails. This is important because it very strongly
indicates that the above inclusion is
strict:  $n {\rm oa}\subset(n+1){\rm oa}$.

We obtained these lattices by analyzing Kochen-Specker sets. The Kochen-Specker
sets correspond to strictly quantum systems. They cannot be given
a classical interpretation at all, and therefore their lattice representation
should be a proper lattice representation.  We wanted to find out how they
can be constructed, what they look like, and what their Hasse diagrams
look like.  In the literature, we only found that 3-dim KS systems in
particular
and spin-1 systems in general were described by Greechie diagrams (as we
stressed in the Introduction).

To our surprise, all but Peres' Greechie lattices violated 3OA, and to our joy
Peres' Greechie lattice passed 3OA through 6OA but violated 7OA. We say
surprise, because every lattice of a quantum system must be represented
by a sublattice of a Hilbert lattice, and the
violations of OAs meant that the representation by Greechie lattices
is incorrect.
It is incorrect because the Greechie lattices are not sublattices
of the lattice of closed subspaces of a Hilbert
space, a fact that escaped the authors mentioned in the Introduction.
Therefore in Sec.\ \ref{sec:represent}, we explain what a proper lattice of any
quantum system should look like and how we can use MMPLs when we need a
lattice for a particular system which passes particular equations.

We should mention out that the numbers of elements (atoms and blocks)
of the smallest known
lattices that satisfy $n$OA but violate $(n+1)$OA do not grow
exponentially. For $3\leq n \leq 7$ we have 13, 17, 22, 28, 33 and
7, 10, 13, 18, 21 atoms and blocks, respectively.\
\cite{pm-ql-l-hql2}  An important open question is whether there
is a pattern that can be identified in this or a similar series
of lattices.  If so, that might lead to a proof that  $(n+1)$OA is
strictly stronger than $n$OA for all $n$.

Since the class of Hilbert lattices (HL) is a subclass of $n$oa for all
$n$ (as Th.\ \ref{th:hs-ssnoa} shows), an open question is what
additional conditions must be added to $n$OA to specify HL, for both the
finite and infinite dimensional cases?  Are there other classes of
equations that hold in every HL when we do not introduce states on it?
(The other known equations such as Godowski's and Mayet's
\cite{pm-ql-l-hql2} assume states.)  How far can we define HL only by
means of sets of equations added to an OL?

\subsection*{Acknowledgment}
One of us (M. P.) would like to thank his host Hossein Sadeghpour
for support during his stay at ITAMP.

Supported by the US National Science Foundation through a
grant for the Institute for Theoretical Atomic, Molecular,
and Optical Physics (ITAMP) at Harvard University and Smithsonian
Astrophysical Observatory and the Ministry of
Science, Education, and Sport of Croatia through the project
No.\ 082-0982562-3160.

Computational support was provided by the cluster Isabella of
the University Computing Centre of the University of Zagreb and
by the Croatian National Grid Infrastructure.

\end{document}